\documentclass[journal]{IEEEtran}
\usepackage{cite}
\ifCLASSINFOpdf
\else
\fi
\usepackage{amsmath}
\usepackage{pdfpages}
\usepackage{amsthm}
\usepackage{amssymb}
\usepackage{array}
\usepackage{url}
\usepackage[ruled,vlined]{algorithm2e}
\usepackage{graphicx}
\usepackage{flushend}
\usepackage[caption=false,font=footnotesize]{subfig}
\usepackage{multirow}
\usepackage{mathrsfs}  
\usepackage{xcolor}
\newtheorem{lemma}{Lemma}
\newtheorem{theorem}{Theorem}
\newtheorem{mydef}{Definition}
\newtheorem{remark}{Remark}

\newtheorem{example}{Example}
\usepackage{xfrac}
\usepackage{enumitem}

\begin{document}

\title{Scalable Secret Key Generation for \\ Wireless Sensor Networks}

%
%

\author{Ufuk~Altun, Semiha T. Basaran, Gunes Karabulut~Kurt,~\IEEEmembership{Senior~Member,~IEEE,} Enver~Ozdemir,~\IEEEmembership{Senior~Member,~IEEE} 
\thanks{Ufuk~Altun  is with the Department of Electrical and Electronics  Engineering, Koc University, Istanbul, Turkey, email: ualtun20@ku.edu.tr.  During this work, Ufuk~Altun and Gunes Karabulut Kurt were with the Department of Electronics and Communications Engineering, Istanbul Technical University, Istanbul, Turkey.} 
      
\thanks{Semiha T. Basaran is  with the Department of Electronics and Communications Engineering, Istanbul Technical University, Istanbul, Turkey, email: tedik@itu.edu.tr.} 
\thanks{G. Karabulut Kurt is with the {Poly-Grames Research Center, Department of Electrical Engineering,  Polytechnique Montr\'eal, Montr\'eal, Canada, e-mail: gunes.kurt@polymtl.ca}.}
\thanks{Enver Ozdemir is with Informatics Institute, Istanbul Technical University, Istanbul, Turkey e-mail:  ozdemiren@itu.edu.tr.}
\thanks{This work was supported in part by NSERC Discovery Grant.}
}

\IEEEoverridecommandlockouts 

\maketitle

	\begin{abstract}
	Future sensor networks require energy and bandwidth efficient designs to support the growing number of nodes. The security aspect is often neglected due to the extra computational burden imposed on the sensor nodes. In this paper, we propose a secret key generation method for wireless sensor networks by using the physical layer features. This key generation method is based  on the superposition property of wireless channels. The proposed method exploits the multiple access property of the wireless channel with simultaneous transmissions as in the \textit{analog function computation} technique to solve the latency and scarce bandwidth problems of highly populated dense networks. All nodes use the same time and frequency block to provide scalability that is \textit{linearly proportional} to the number of nodes. The proposed method also benefits from the network density to provide security against eavesdroppers that aim to sniff the secret key from the channel. The security of the proposed method against eavesdroppers is analytically studied. Moreover, their application in multiple layers is investigated. The presented results have shown that there is a trade-off between the total power consumption and total used bandwidth for secret key generation. Lastly, the error probability of the generated keys due to thermal noise and channel estimation error is investigated with computer simulations and compared with broadcasting-based benchmark model.
	\end{abstract}
	
	

\IEEEpeerreviewmaketitle

	\section{Introduction}

\IEEEPARstart{T}{he} number of nodes in wireless sensor networks (WSNs) that is connected with the Internet of things (IoT) is expected to grow significantly in the future. For this reason, communication technologies that is able to handle the limited wireless resources more efficiently than the traditional methods are necessary for IoT. Providing security in these settings is a challenging task that costs network resources proportional to the network size with conventional security measures. A common method for hiding information from adversaries is the encryption and decryption of the messages at the upper layers. However, providing a secret key to nodes is a vulnerable process that can be intercepted by the eavesdropper and also requires the existence of a central node, which improves the network cost. 

As an alternative strategy, physical layer key generation (PLKG) has gained great attention in recent years. PLKG algorithms are mainly based on exploiting the randomness and uniqueness of the wireless channel to keep eavesdroppers in dark. Most of the PLKG algorithms are designed for two node networks as in \cite{Lu2021, Hu2021, Furqan2021, Zhang2021}, where only pairwise communication is possible. As a result, scalability and resource management problems are generally overlooked. Especially, implementing a pairwise algorithm consecutively in a large network is time, bandwidth and energy inefficient as well as highly vulnerable to eavesdroppers.

Several studies in the literature such as \cite{Xu2016a, Xu2020, jiao2020, Peng2020, Liu2014, Xiao2018, Thai2019, Guyue2019, Zhang2019,9697095 } considered this problem and focused on group networks. These studies can be classified into two categories. In the first category \cite{Xu2016a, Xu2020, jiao2020, Peng2020}, the algorithms include pairwise communications to generate a group secret key. These studies target relatively small networks and require impractical communication times in dense networks. In the second category \cite{Liu2014, Xiao2018, Thai2019, Guyue2019, Zhang2019 }, the methods are purely based on broadcasting and are suitable for dense networks. In \cite{Liu2014}, the authors consider a new metric where, instead of raw RSSI values, difference of two received signal strength indicator (RSSI) values from two different channels is used. Other than the secrecy improvement of the new metric, study also considers the scenario where all nodes are not in each other's communication range. In \cite{Xiao2018}, the network consists of a central node and a reference node to organize a scalable key distribution model. Each node uses the channel information between itself and a reference node for key reconciliation. 

The authors of \cite{Thai2019} consider the multiple antenna scenario where the method depends on multiple antenna transmission of nodes. The proposed algorithm solves the optimization problems of antenna selection and time scheduling with brute force search. A scalable lightweight algorithm is proposed in \cite{Guyue2019}. The method improves its scalability by dropping the reconciliation phase and relying the similarity between pairwise bit sequences. As main drawbacks, these algorithms either assume a noiseless channel or apply channel coding to provide a noiseless channel between the nodes for information exchange. In \cite{Zhang2019}, an orthogonal frequency-division multiple access (OFDMA)-based key generation model is proposed. The method aims to reduce the duration of channel estimation phase by exploiting OFDMA. Specifically, the method is able to complete the channel estimation process in a single time slot by assigning a unique frequency to each node. Since OFDMA networks already assign a specific frequency to each user, the proposed method does not require additional bandwidth. However, the idea is only limited to OFDMA systems. Recently, multiple antenna scenario is considered in \cite{9697095}. The study implements a version of index modulation where the common information is encoded to the index of the non-activated antenna. Also, the authors present a thorough theoretical investigation. 

We consider the key generation problem of multiple node networks, where the scalability and efficiency are essential requirements. The methods given in \cite{Xu2016a, Xu2020, jiao2020, Peng2020} fail to satisfy these requirements since they rely on pairwise communications and their time/bandwidth consumption increases exponentially with the network size. More related studies such as \cite{Liu2014, Xiao2018, Thai2019, Guyue2019, Zhang2019, 9697095} are able to satisfy these requirements since their communication model is based on broadcasting. However, broadcasting a key component is essentially vulnerable to eavesdroppers when additional countermeasures are not taken. Most of these studies require third parties, reference nodes (requires doubled channel estimation process) or perfect feedback channels (requires channel coding overhead). In Table \ref{tab:literature}, we present various group key generation methods with their pros and cons. Our model significantly differs from these methods on the communication model by exploiting simultaneous transmission instead of broadcasting.

\begin{table*}[t]
    \centering
    \caption{Related studies in the literature.}
    \begin{tabular}{|p{1.9cm}|p{7cm}|p{7.5cm}|} \hline The study & Advantages & Disadvantages \\ \hline
 Liu \textit{et al.} \cite{Liu2014} & Considers the scenario where all nodes are not in their communication range. & A central and a reference node is required which requires doubled channel estimation process. Requires secrecy amplification process to overcome the vulnerability of broadcasting information. \\ \hline

Xiao \textit{et al.} \cite{Xiao2018} & Comparative results with \cite{Liu2014} is presented. Attains higher \textit{key length rates} than \cite{Liu2014}. & A central and a reference node is required which doubles the channel estimation process. Requires a noiseless channel for key generation. \\ \hline

Thai \textit{et al.}    \cite{Thai2019} & Presents a successful algorithm for multiple antenna networks. Includes real testbed implementation. & Only applicable to networks where each node is equipped with multiple antennas. \\ \hline        

Li \textit{et al.}  \cite{Guyue2019} & Comparative results with \cite{Liu2014} is presented. Attains higher \textit{bit error rates} than \cite{Liu2014}. Lightweight since key reconciliation process is omitted. & Noiseless broadcast channels are required. \\ \hline

Zhang \textit{et al.} \cite{Zhang2019} & Reduces the load of channel estimation process for OFDMA networks.  & The contribution is only valid for OFDMA networks. \\ \hline
    \end{tabular}
    \label{tab:literature}
\end{table*}

Simultaneous transmission indicates transmission of multiple nodes at the same time and frequency interval. Although it is strikingly time and bandwidth efficient (since all nodes use the same time and frequency slots), its application is very limited in the literature. The reason simply comes from the fact that individual signals superimpose over the channel and can not be reconstructed at the receiver. However, various real world tasks are not interested in individual information and the superimposed signal is enough to work with. As an example, the task of reliably measuring room temperature with multiple sensor nodes can be given. Since a user is only interested in the average sensor measurements, the nodes can simultaneously transmit their readings to a center and the superimposed signal is enough to compute the average temperature. This example is the main target of analog function computation (AFC) studies, pioneered by Gastpar~\cite{Gastpar2003, Nazer2007a} and Sta\'nczak~\cite{Goldenbaum2009a, Goldenbaum2013, Goldenbaum2013b}. Essentially, AFC aims to combine communication and computation processes of a given task over the wireless multiple access channel (W-MAC). From a mathematical view, W-MAC naturally computes a summation operation over the simultaneously transmitted signals. AFC proposes that this paradigm can be extended to compute any function by using proposer \textit{pre} and \textit{post}-processing steps. An overview of AFC applications is given in \cite{altun2021magic}.

In~\cite{zhu2018mimo} and~\cite{chen2018over}, the idea behind the AFC is extended to MIMO and beamforming technologies. 
In~\cite{Basaran2020}, an energy efficient AFC scheme is proposed for densely deployed IoT networks. In \cite{Wang2021}, AFC is implemented with intelligent reflecting surfaces. In \cite{Liu2021}, AFC is exploited for the security of federated learning algorithms. In our recent study \cite{altun2020}, we have considered exploiting AFC to provide an authenticated data transmission model against active attacks. To the best of our knowledge, this is the first study that exploits AFC and simultaneous transmission to generate group secret key. 

In this paper, we focus on the key generation problem of WSNs which presents two major design challenges: scalability for dense networks and limited communication resources (e.g., time, bandwidth). Considering these challenges, our motivation is to provide a flexible and scalable group key generation method which can be applied to dense and resource limited WSNs. The study differs from existing literature by its communication model which is based on simultaneous transmission of all nodes at the same time and frequency (contrary to existing studies which all nodes sequentially broadcasts their key components). Our contributions are listed as follows.

    \begin{itemize}
       \item The key agreement time of the proposed model is linearly proportional to the network size as a result of the unique communication model, i.e., $N$ nodes can agree on a secret key in $N$ time slots.
       \item The method is secure against eavesdroppers. The security aspect increases with the network size, since eavesdropped signals are the superposition of multiple signals coming from different paths. Each unique path eventually contributes to the distortion of Eve's received signal. 
       \item We extend our algorithm to be used in multiple layers by decomposing the WSN into multiple subnetworks. The decomposition provides a trade-off between energy, time and bandwidth consumption. Hence, the method is flexible to the needs of WSN.
       \item The proposed method does not require a trusted third party, a reference node or a central node to distribute information to other nodes. As a result, it is energy and complexity efficient on channel estimation or key agreement processes.

    \end{itemize}

    The rest of this paper is organized as follows. In the next section, background knowledge needed to present our method is given with definitions. In Section III, we present our key generation approach. The security aspect of the proposed approach is investigated analytically and supported with simulations in Section IV. Lastly, the paper is concluded in Section V.

	\section{Preliminaries}

	\begin{figure}[t]
		\includegraphics[width=\linewidth]{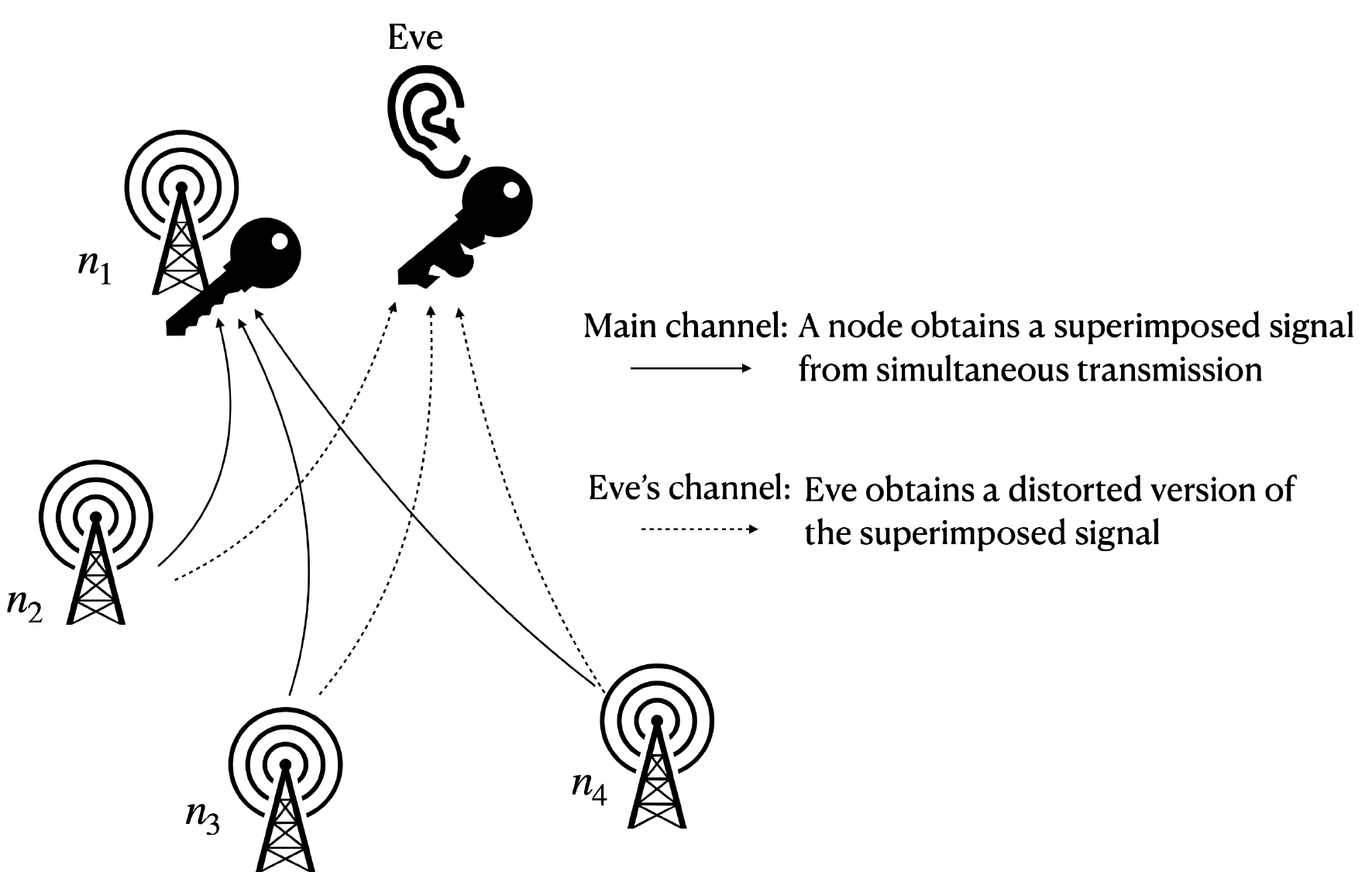}
		\caption{Illustration of the network model with nodes, $n_1, ..., n_4$ and Eve. The proposed key generation approach is based on simultaneous transmission of all nodes. A node receives a superimposed signal which contains information from the remaining nodes and can use this information to generate secret key. On the other hand, Eve obtains a highly distorted version of the superimposed signal since Eve's channel is different (ideally independent) than the main channels.}
		\label{basemod}
	\end{figure}

Consider a wireless sensor network that consists of $N$ nodes as depicted in Fig.~\ref{basemod}. We assume that an eavesdropper (Eve) is also present and monitors the transmitted signals. We denote the set of all nodes by $\mathcal{N} = \{n_1,\ldots,n_N\}$, where nodes are in an arbitrary order. We specify any transmitting node as $n_i \in \mathcal{N}$ and any receiving node as  $n_j \in \mathcal{N}$, where $i,j \in \{1,\ldots,N\}$. An overview of used notations in this paper is presented in Table
~\ref{tab:notations}. Also, we would like to address some important notions in detail as follows.
\begin{itemize}
    \item \textit{Simultaneous transmission:} Nodes in the network transmit their signals simultaneously during the same time and frequency block.
    \item \textit{Half-duplex communication}: Nodes in the network can only transmit or receive at a time block.
    \item \textit{The set of key components ($\mathcal{P}$):} A publicly known set that consists of Gaussian prime integers, $\mathcal{P}\subseteq \mathbb{N}$. 
    \item \textit{Key component/prime component ($p_i$):} Each node, $n_i$, randomly selects an element from the set of key components as $p_i$ in order to use in the proposed algorithms. $p_i$ is assumed to be only known by $n_i$.
    \item \textit{Common information ($S$):} An information that is shared by all legitimate nodes as the output of the proposed algorithms. $S$ is a function of the key components. Here, $S$ can be considered as a secret key as well as it can be later processed to generate new secret keys.
\end{itemize}

\begin{table}[t]
    \centering
    \caption{A list of notations.}
    \begin{tabular}{|p{.8cm}|p{6.5cm}|} \hline Notation & Definition \\ \hline
    
 $\mathcal{N}$ & The set of all nodes.  \\ \hline

 $\mathcal{P}$ & The set of key components.  \\ \hline

 $\mathcal{S}$ & The set of nodes in a subnetwork.  \\ \hline

 $N$ & Number of nodes, i.e., network size.  \\ \hline

 $K$ & Number of prime factors of $N$, i.e., number of dimensions after decomposition.  \\ \hline

 $n$ & A node.  \\ \hline

 $i$ & Index of a transmitting node.  \\ \hline

 $j$ & Index of a receiving node.  \\ \hline

 $E$ & Index of eavesdropper, i.e., Eve.  \\ \hline

 $y$ & The received signal.  \\ \hline

 $h$ & The channel fading coefficient.  \\ \hline

 $\omega$ & The additive white Gaussian noise.  \\ \hline

 $x$ & The transmitted signal.  \\ \hline

 $\varphi$ & The \textit{pre}-processing function.   \\ \hline

 $\psi$ & The \textit{post}-processing function.  \\ \hline

 $C$ & The prime factors of a decomposed network.  \\ \hline

 $TPC$ & Total power consumption. \\ \hline

 $TCT$ & Total communication time.  \\ \hline

 $TOB$ & Total occupied bandwidth.  \\ \hline

    \end{tabular}
    \label{tab:notations}
\end{table}

Our purpose is to provide a common information $S$ to all legitimate nodes using in-network communications without leaking $S$ to Eve. 
In this system we consider the following \textit{attacker model}. An eavesdropper can gain access to capture the ongoing transmissions. Considering a highly capable eavesdropper, we assume that she can estimate the channel between the transmitter and herself ideally. As an improvement to attacker models in the literature, we consider following two different attacker models to test our method against.
  \begin{enumerate}[label=Attacker-\arabic*),leftmargin=*]
   \item Eve's spatial location is highly correlated with the legitimate receiver, i.e., Eve's (perfectly estimated) channel coefficients are distorted versions of legitimate receivers. 
    \item Eve's spatial location is identical to the legitimate receiver, i.e., Eve's (perfectly estimated) channel coefficients are identical to the legitimate receiver.
\end{enumerate}

Note that as ideal channel estimates can not be obtained in actual communication systems, Eve has been equipped with powerful features that are not necessarily realistic. The communications that generate a common information takes place over the wireless multiple access channel which is defined as follows.
    \begin{mydef}[W-MAC~\cite{Goldenbaum2013}]\label{W-MAC}
	Let $x_i$ be the signal transmitted by the user $n_i$ and let $y_j$ be the observation of W-MAC at the receiver end of the $n_j$. The channel fading coefficient and the additive white  Gaussian noise (AWGN) between the $n_i$ and $n_j$ are denoted by $h_{ij}$ and $\omega_j$, respectively. Then, we define the signal received by $n_j$ over the W-MAC as
		\begin{equation}\label{MAC_def}
			y_j=\sum_{{
   			i=1
  			}} ^{N} h_{ij}x_i + \omega_j.
		\end{equation}
	\end{mydef}
	
The authors in~\cite{Goldenbaum2013b} prove that with proper adjustments to signals at transmitter and receiver ends, any desired function is computable over the W-MAC. The adjustments named as \textit{pre}-processing and \textit{post}-processing functions of transmitter and receiver ends, respectively are defined as follows.
	\begin{mydef}[\textit{pre}-processing function~\cite{Goldenbaum2013b}]\label{pre_def}
A process $\varphi_i$ defined as $x \in \mathbb{R},  \varphi_i(x_i)=(\varphi\circ x_i)$ is the \textit{pre}-processing function of the transmitter node $n_i$.
	\end{mydef}
	\begin{mydef}[\textit{post}-processing function~\cite{Goldenbaum2013b}]\label{post_def}
A process $\psi_j$ defined as $x \in \mathbb{R}, \psi_j(y_j)=(\psi_j\circ y_j)$ is the \textit{post}-processing function of the receiver node $n_j$.

\end{mydef}
Basically, \textit{pre}-processing and \textit{post}-processing functions provide the means to adjust the channel model to represent a desired function. Multiplication of the signals over the W-MAC is the main concept that our algorithms are built on. The \textit{pre}-processing function of the multiplication operation can be given as
\begin{equation}\label{pre_mult}
	{\varphi}_i(x_i) = \dfrac{1}{h_{ij}}\ln(x_i).
\end{equation}\par
    
\begin{remark}
In AFC applications, channel is assumed to be noiseless ($\omega$=0) and inverted beforehand with the channel state information (CSI) ($h$=1). As a result, the channel fading coefficient $h_{ij}$ is not given in the \textit{pre}-processing functions. 
Also, we consider a more realistic scenario, where AWGN is added to the received signals.   
\end{remark}
	
Natural property of logarithm enables us to obtain the logarithm based product of signals. Accordingly, conversion to polynomial base is required at the receiver end. The \textit{post}-processing function, 
	\begin{equation}\label{post_mult}
		{\psi}_j(y_j) = \textrm{exp}[y_j],
	\end{equation}
completes the reconstruction of transmitted components at the the receiver node. With the \textit{pre}-processing function in~\eqref{pre_mult} and the \textit{post}-processing function in~\eqref{post_mult}, product of transmitted components can be obtained at the receiver node as follows
	\begin{equation} \label{explain}
    \begin{split}
    	\psi & \left(\sum_{i=1} ^{N} \varphi_i(x_i) \right)
  		=\psi\left(\ln\left(\prod_{{i=1}} ^{N} x_i\right) \right)
  		=\prod_{{i=1}} ^{N} x_i. 
  	\end{split}
    \end{equation}
In the following section, based on our observations from AFC, especially the multiplication operation over the wireless channel, we generate a secret key that is shared by all WSN nodes.


\section{ Secret Key Generation}
	
The key generation method is based on the idea that a single node can obtain a function of information from other nodes in one simultaneous communication. Repeating the communication for $N$ nodes results in a network where all nodes possess the same information. If the common information is securely generated by all nodes, it can be used as a secret key in any encryption-based communication.

As in all AFC applications, receiving a function of information from the W-MAC requires signal processing, which is given in Def.~\ref{pre_def} and~\ref{post_def}. Application of the signal processing eventually matches the W-MAC with the desired function.  We firstly define the matched channel model for half-duplex communication (W-HMAC). Then, the \textit{pre} and \textit{post}-processing functions that are essential for the key generation are proposed. 

	\begin{mydef}[W-HMAC]\label{W-HMAC_def}
Let $\varphi_i(x_i)$ be the \textit{pre}-processing function applied by the transmitting node $n_i$ and $\psi_j(y_j)$ be the \textit{post}-processing function applied by the receiving node $n_j$ where $i,j \leq N, i \neq j$. Then, we define 
		\begin{equation}
			\psi_j(y_j)=\psi_j\left(\sum_{\substack{i=1,i\neq j}} ^{N}	h_{ij}\varphi_i(x_i)\right),
		\end{equation}
the \textit{post}-processed output of W-MAC as W-HMAC. 
	\end{mydef}

An illustration of the W-HMAC is given in Fig.~\ref{whmac}. In this model, a receiver node ($n_j$) obtains a function of the components that are transmitted from the remaining nodes ($n_i\in\mathcal{N}\setminus\{n_j\}$). The same common information can be delivered to all nodes by repeating the communication over W-HMAC and iterating the receiver node. However, providing the same common information to all nodes requires a \textit{pre}-processing function that inverts the channel's fading coefficients as,
 \begin{equation} \label{pre}
     \varphi_i(p_i) = \dfrac{\ln{p_i}}{h_{ij}}.
 \end{equation} 
Recall that \eqref{pre} contains channel coefficients, which can be obtained with a channel estimation process. Since transmitting nodes ($n_i$) require their channel estimates only towards $n_j$, broadcasting a pilot signal from $n_j$ is sufficient to provide all other nodes with $h_{ij}$. 

\begin{remark}
The channel estimation process can be vulnerable to active attacks (jamming or spoofing) since a malicious node can interfere with the pilot signal. Without a countermeasure, this scenario would lead to high discrepancies in the generated secret keys, and practically block the key generation process. Note that an active attack on pilot signals can be modeled as channel estimation error at the receiver. A numerical investigation on channel estimation error can be found in Section \ref{sec:numerical}. Specifically, Fig.~\ref{errprob3} depicts the impact of channel estimation error on the error probability of the proposed method. The simulation scenario considers addition of unwanted noise to the received signal to imitate the imperfections of the real world. As the noise power increases (transmit power of the attacker increases), the error in the generated keys increases.
\end{remark}

Any communication over the W-HMAC always outputs the same function. However, the inputs of W-HMAC are different at each iteration, since the receiver node index ($j$) changes. By applying the \textit{post}-processing function
\begin{equation} \label{post}
	{\psi}_j(x_j,y_j) = x_j\textrm{exp}[y_j],
\end{equation}
the receiver node includes its own key component to the received signal and obtains the common information.  
It should be noted that the \textit{post}-processing function requires the prime component of the receiver node which is only available to the legitimate receiver. Consequently, eavesdropping $y_j$ from W-HMAC is not enough to obtain ${\psi}_j(x_j,y_j)$. The proposed half-duplex communication approach is summarized in Algorithm~\ref{alg1}. 

	\begin{figure}[t]
		\includegraphics[width=\linewidth]{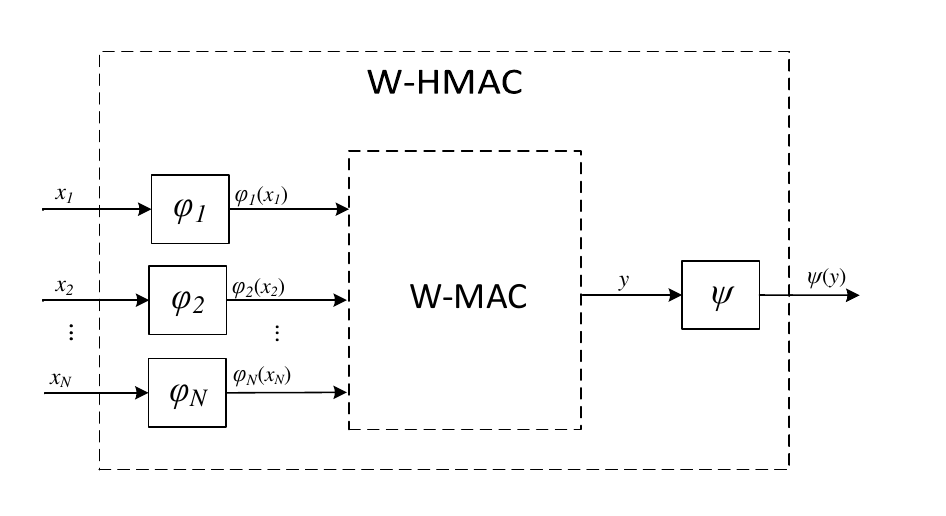}
		\caption{Illustration of W-HMAC where $x_1,...,x_n$ are inputs and $\psi(y)$ is output. W-HMAC is half duplex adaptation of W-MAC with proper \textit{pre}-processing and \textit{post}-processing functions to provide a common information to one node.}
		\label{whmac}
	\end{figure}

\begin{algorithm}[t!] \caption{Common information generation over W-HMAC.} \label{alg1}
\SetAlgoLined
\KwResult{All nodes in the WSN obtain the common information as the product of the prime components 
    $ \prod_{
   		i=1
  		} ^{N} p_i.$}
 \textbf{Initialization:} All nodes, $\mathcal{N}$, in the network choose a Gaussian prime integer $p_i\in \mathcal{P}$ as a key component. Set the receiver node index $j=1$.\\
 \Repeat -
 1) The nodes $n_i\in\mathcal{N}\setminus \{n_j\}$ transmit the output of their \textit{pre}-processing function given in~\eqref{pre}.\\
 2) $n_j$ calculates the common information by using~\eqref{post}.\\
 3) Update $j=j+1$.
 \end{algorithm} 
 
 Nodes of the WSN communicate over the W-HMAC where an arbitrarily chosen $n_j\in\mathcal{N}$ obtains the product of the prime components given as 
    \begin{equation}\label{half_output}
		\psi(p_j,y_j)= \prod_{
   		i=1
  		} ^{N} p_i.
	\end{equation}
After the \textit{post}-processing function, $y_j$ obtains~\eqref{half_output} which contains the prime components of all nodes. This common information can be used as a secret key in any encryption process or can be used to generate complex keys with a pre-determined key generation functions, i.e., a function can transform the product of the prime components into a secret key with desired bit length.
 
\begin{remark}
Providing common information to $N$ number of nodes requires $N$ sequential communications over the W-HMAC where each node becomes the receiver node once. Without any assumptions, iterating the receiver node index ($j$) for all WSN requires all nodes to possess an index. It should be noted indexing all nodes is an inexpensive process which can be completed by a single broadcasting before initiating Algorithm \ref{alg1}. As a result, the relationship between the time delay and the number of nodes is in linear scale as opposed to traditional pairwise secret generation methods. Also, the proposed algorithm is highly efficient in terms of bandwidth and energy since the nodes occupy the same frequency and time slot by simultaneously transmitting.
\end{remark}


\subsection{Security of the Secret Key Generation over W-HMAC}

The security of the key generation method relies on the size of the WSN and the fading of the W-MAC that is unique to the path between any two points. As given in the \textit{pre} and \textit{post}-processing functions, the components of the shared key are carried at the amplitudes of the transmitted signals. As a result of this analogous approach, we argue that the proposed method is also robust against eavesdropper attacks that aim to obtain the shared key. In this subsection, we investigate the feasibility of the proposed approach against active and passive attack models.

\subsubsection{Against Passive Attacks}

Multiple access nature of the wireless channel makes passive attacks inevitable. In a passive attack scenario, the attacker illegally eavesdrops the channel and obtains legitimate information. For a key generation scheme, passive attacks present a critical vulnerability since they may lead to the leakage of the secret key. The danger of a passive attack is highly related to the attacker's channel capacity. Previously, we have proposed Attacker-1 and Attacker-2 models to investigate various scenarios. The Attacker-1 model considers a realistic scenario where Eve's location is close to the legitimate receiver and Eve can estimate the channel coefficients of the legitimate path with a small discrepancy. The Attacker-2 model presents a more challenging scenario where Eve can perfectly estimate the channel coefficients of the legitimate path without any error. The resistance of our model against these two models will be examined in the following discussion. 

In the proposed model, communication over W-HMAC provides $n_j$ with the prime components of other nodes. After applying the \textit{post}-processing function, $n_j$ obtains these components from the received signal and multiplies it with its own prime component. Then, $n_j$ can use the obtained common information to generate secret key ($S$) and use it for any encryption and decryption processes. 
The only way for Eve to obtain the secret key is to obtain the prime components at the communication step by sniffing signals correctly towards $n_j$. However, for both attack scenarios, obtaining the output of W-HMAC is not simply sufficient to construct $S$, since Eve still needs the prime component of $n_j$. The only solution for Eve to obtain all the information (prime components) to construct $S$ is to eavesdrop at least 2 rounds of communications where minimum 2 nodes in $\mathcal{N}$ receive output sequentially from W-HMAC. 

In this case, Attacker-1 receives the following signal from one sniffing,
    \begin{equation*}
    \begin{split}
		y=\sum_{\substack{i=1,\;i\neq j}}^N r_{ij}\ln (p_i)
	\end{split}
    \end{equation*}
where $r_{ij}=\sfrac{h_{iE}}{h_{ij}}$. After the \textit{post}-process, Eve obtains $$\psi_E\left(\sum_{\substack{i=1,\;i\neq j}}^N r_{ij}\ln (p_i)\right) =\prod_{\substack{i=1,\;i\neq j}}^N p_i^{r_{ij}}$$
as the output of W-HMAC. Note that the quantity Attacker-1 obtained contains exponentially distorted versions of the prime components which brings us to the following theorem. 
	\begin{theorem} \label{theorem3}
Let $\psi_j(y_j)$ be the output $n_j$ receives from W-HMAC and $r_{ij}=h_{ij}/h_{iE}$. As long as $|1-r_{ij}|\neq 0$, the output Eve obtains from the W-HMAC, $\psi_E(y_E)$, is different from the output $n_j$ obtains. 
	\end{theorem}
	\begin{proof}
Difference between the W-HMAC outputs of $n_j$ and Eve is the following expression,
\begin{equation}
\begin{split}
|\psi_j(y_j)- \psi_E(y_E)|&= |x_j\exp [y_j]- \exp[y_E]|\\
&=\prod_{\substack{i=1}}^N p_i-\prod_{\substack{i=1,\;i\neq j}}^N p^{r_{ij}} \\
\end{split}
\end{equation}
which can be organized as
$$\prod_{\substack{i=1}}^N p_i \underbrace{\left( 1- \prod_{\substack{i=1,\;i\neq j}}^N p_i^{r_{ij}-1}\right)}_{E_r}.$$
	Note that Eve will have a discrepancy once the error, $E_r$, is not 0 and it occurs when $|1-r_{ij}|$ is not equal to $0$.
	\end{proof}
Depending on the channel fading coefficient, choosing high digit prime integers makes decipher of $S$ infeasible as a high discrepancy occurs between the W-HMAC output of Eve and $n_j$. The following lemma and example given below illustrate the destructive effect of the fading and simultaneous transmission on Eve's computation. 
	\begin{lemma}\label{lemma1}
	Let $s=1.a_1a_2a_3\dots a_m$ be a real number where $m$ represents the length of the decimal part such that the first non-zero term in the decimal part is $a_r$ for some positive integer $r$. Then $n\cdot s$ has the same first $r$ digits with $n$ and the remainders will be different.
	\end{lemma}
	\begin{proof}
The corresponding proof is given in Appendix A.
	\end{proof}
	\begin{example}\label{example1}
Suppose each node has selected a prime integer with at least 6 digits and $E$ obtains $\psi_E(y_E)$. Let $|r_{ij}|>1.0001$ then from Lemma~\ref{lemma1}, last two digits of the prime numbers which was selected by $n_i$ will change and the multiplication $\psi_E(y_E)$ will be completely different (except maybe first two digits) than $\psi_j(y_j)$.
	\end{example}
An observation from Example~\ref{example1} is the impact of the number of nodes. Since each input of $n_i$ will be distorted by $|r_{ij}|$, a linear increment in the number of nodes will exponentially increase the discrepancy that Eve faces.   

We can observe that a small discrepancy in Eve's estimation of the legitimate channel can prevent Attacker-1 from ever reconstructing the secret key. For the Attacker-2 scenario, we assume that Eve can estimate the legitimate channel without any error. Although this assumption is not realistic, it also requires Eve to be spatially located in more than one location. Since a receiving node uses its own key component to construct the secret key, what Eve receives from the channel does not contain the prime component of the receiver. Hence, Eve has to listen to at least another communication to decipher the missing key component. Moreover, Eve has to be located at the exact location of the receiving node to perfectly estimate the channel coefficients. As a result, giving perfect estimation capabilities to Eve is not enough for it to decipher the secret key; it also requires the capability to eavesdrop multiple locations sequentially. Even in this scenario, Eve faces more errors than legitimate nodes. The reason is that a legitimate node constructs the key with a single signal reception which introduces a single noise component into the key. However, Eve constructs the key by sniffing multiple communications which encompass multiple noise components to its key.

\subsubsection{Against Active Attacks}
Active attacks are another vulnerability of the wireless channel. In active attack scenarios, the attacker acts invasive and emits its signals into the channel to disrupt legitimate communication. In our system model, active attacks can carry two purposes: spoofing the secret key by implementing a fake key component into the system or disrupting the key generation process by jamming the channel. The existing countermeasures against active attacks usually rely on detecting the attack (authenticating the legitimate nodes) or avoiding the attacker \cite{8792139}.

The general countermeasures for spoofing attacks are based on detecting the spoofer. In \cite{altun2020}, we have considered a simultaneous transmission scenario and investigated its feasibility against spoofing attacks. Our results showed that a spoofer has to perfectly estimate its channel and the legitimate channel coefficients to carry a successful spoofing attack. Otherwise, the receiving node is able to detect the error. Moreover, it is shown that the receiving node is able to classify the error as a natural noise or a spoofing attack with high probability. In our key generation model, the results of \cite{altun2020} are still viable with the addition that the spoofer also needs to repeat its attack for all nodes to insert a fake key into the system successfully. 

The spoofers can also target a single node to prevent it from obtaining the correct key. In this scenario, the spoofer does not need to repeat its attack for all nodes. Although the proposed method in \cite{altun2020} can detect spoofers with high probabilities, we  can also consider other countermeasures against spoofing attacks. First of all, a key reconciliation step can be added to verify that every node generates the same key. A reconciliation step aims to detect and correct the errors between the generated keys~\cite{HUTH201684}. For this purpose, nodes can discuss over a public channel (which may, unfortunately, leak partial information to eavesdroppers) and compare their keys. In case of an active attack, a reconciliation step can detect disagreements between the nodes. Hence, the nodes can repeat the proposed method on different frequencies until the attacked frequencies are avoided. However, adding a key reconciliation system without leaking the key to eavesdroppers is a resource inefficient task. In fact, the methods in literature~\cite{HUTH201684} significantly increase the communication overhead. The simultaneous transmission-based methods promise to reduce the communication burden in various areas \cite{altun2021magic}. It is a future direction to propose a lightweight key reconciliation process suitable for multi-node networks.

In addition to key reconciliation methods, other physical layer authentication systems can be implemented at both the pilot transmission and communication stages to detect active attacks. Physical layer authentication techniques fundamentally exploit channel-based features to authenticate untrusted nodes. Compared to traditional methods, physical layer-based methods can provide lightweight applications and relieve the burden of key sharing/management~\cite{9279294}. Moreover, recent advancements in this area made physical layer authentication a promising countermeasure against active attacks. For instance, deep learning-aided methods as in~\cite{Qiu2020_new, Liao2019a} can be adopted to the proposed approach to detect spoofed signals with high performance.

The jamming attacks pose another threat to the proposed key generation model. The general countermeasures against jamming attacks are based on frequency hopping and spread spectrum techniques. These methods require additional steps before communication, increasing complexity and reducing time efficiency. Contrary to key reconciliation processes, the communication overhead of anti-jamming techniques is independent of the number of nodes and applicable for large networks. Although we see no restraint in implementing existing anti-jamming models (as in \cite{7996660}) into the proposed model, it is beyond the scope of this paper to investigate the feasibility and performance of such methods. We consider these analyses as a major future direction.

\subsection{Key Generation over Multiple Layers} \label{skgmultilayer}

The proposed channel model, W-HMAC, can be considered as a function that outputs a common information at the receiver node. Since this information consists of the prime components of all nodes, any receiver obtains the same common information over W-HMAC. Note that a communication over W-HMAC produces $S$ in a single node. Intuitively, selecting each node as receiver node once (where the remaining nodes are transmitters) is a simple solution that produces $S$ to all nodes in the WSN. In addition to this concept, various W-HMAC topologies can be constructed by dividing WSN into multiple subnetworks (and dividing the key generation process into multiple layers). As we will discuss in this subsection, using W-HMAC at multiple layers can provide additional flexibility and improve efficiency.

Firstly, we define the following metrics in order to investigate the performance of multi-layer systems. 

\begin{enumerate}
    \item \textit{Total power consumption (TPC):} Power is an important and scarce resource in WSNs. To compare various network models, the total required transmission power is considered as a performance metric. For this purpose, any transmission is assumed to consume unit power. TPC is measured as the total number of transmissions until each node in the WSN produces $S$.
    \item \textit{Total communication time (TCT):} Time is another resource that defines the performance of a key generation method. Contrary to conventional pairwise methods, proposed approach enables simultaneous transmission in which multiple nodes can transmit their information at the same time. Here, we assume that a transmission requires unit time. TCT is defined as the total required unit time until all nodes in the WSN obtains $S$.
    \item \textit{Total occupied bandwidth (TOB):} Communication over W-HMAC requires a single frequency block since it enables simultaneous transmission. We assume that communication over W-HMAC requires only one frequency block. TOB is defined as the total occupied frequency blocks until all nodes in the WSN obtain $S$.
\end{enumerate}

Consider a single layer key generation system where all nodes communicate through a single W-HMAC as given in Algorithm \ref{alg1}. We consider this model as \textit{single layer W-HMAC} with size $N$. TPC of this scenario is $N(N-1)$ where $N-1$ nodes transmit for $N$ iterations. TCT is $N$ unit time since W-HMAC takes a single time slot and is repeated for $N$ iterations. Lastly, TOB is $1$ frequency block since W-HMAC enables simultaneous transmission. Our objective at this stage is to design multi-layer key generation configurations that can reduce TPC and TCT. 

\begin{algorithm}[t] \caption{Common information generation over W-HMAC with multiple layers.} \label{alg2}
\SetAlgoLined
\KwResult{All nodes in the WSN obtain the product of the prime components as 
    $ \prod_{
   		i=1
  		} ^{N} p_i.$}
 \textbf{Initialization:} Decompose $N$ into its prime factors as $N= C_1 C_2 \dots C_K$.\\
 Arbitrarily place all nodes, $\mathcal{N} = \{n_1, n_2, \dots, n_N \}$, into $K$ dimensional space such that a node $n_i$ can be represented as $n_{c_1,c_2,\cdots, c_K}$ where $c_k$ indicates its place in the $k^{th}$ dimension. \\
 Set iteration number $k=1$.\\
 \Repeat -
1) Divide the WSN into $N/C_k$ subnetworks such that, $$ \mathcal{N} = \bigcup_{i=1}^{N/C_k} \mathcal{S}_{i}^{k}, $$ where $\mathcal{S}^k_{i}= \{ n_{c_1,c_2, \cdots, c_K} \}, c_k = i, \forall c_1, c_2 \cdots, c_K $ is the set of nodes in the $i^{th}$ subnetwork.  \\
 2) Each subnetwork initiates Algorithm~\ref{alg1} in itself and obtains a common information. 
 
 3) Each node updates its prime component, $p_i$, as the common information obtained from Algorithm~\ref{alg1}.
 
 4) $k=k+1$.
 \end{algorithm}

For this purpose, we consider decomposing the network into multiple subnetworks and applying Algorithm~\ref{alg1} in multiple layers. The general framework for multi-layer configurations is presented in Algorithm~\ref{alg2}. The algorithm cardinally states the rules on how to design subnetworks (i.e., how to assign nodes into subnetworks) for multiple layers. In each layer, these subnetworks apply Algorithm~\ref{alg1} independently and simultaneously. At the end of a layer, all nodes update its prime component with the output of Algorithm~\ref{alg1}. The layer and subnetwork configuration of Algorithm~\ref{alg2} enables all nodes to obtain the same secret information at the end of the last layer. Details of Algorithm~\ref{alg2} can be visualized with an example given in Fig.~\ref{subnetworks}.


\begin{figure*}[t]
\centering
\includegraphics[width=1\linewidth]{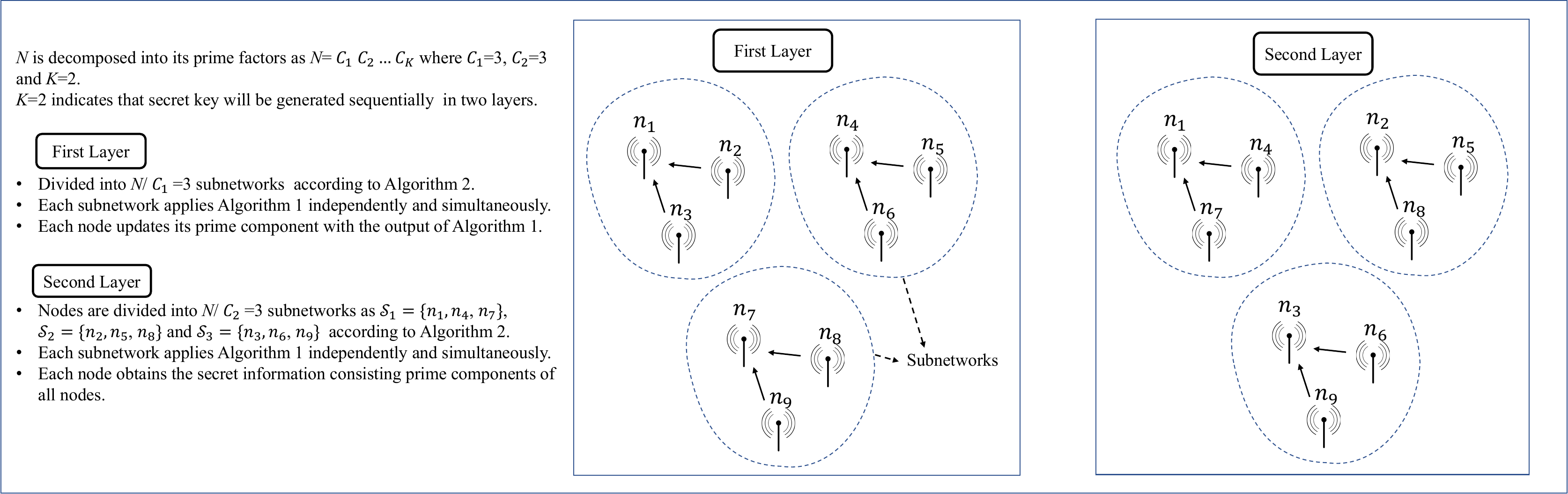}
\caption{An example of multi-layer secret generation model where $N=9$ nodes apply Algorithm~\ref{alg2} and generate a common information in two layers. In each layer, nodes are distributed to subnetworks. }
\label{subnetworks} 
\end{figure*}

A two-layer scenario where $9$ nodes aim to generate $S$ is illustrated in Fig.~\ref{subnetworks}. The key generation process is divided into two layers where WSN is divided into three subnetworks at each layer. In the first layer, each subnetwork uses W-HMAC at different frequencies. Each subnetwork uses Algorithm~\ref{alg1} to generate a common information which takes three time slots. At the end of the first layer, each subnetwork obtains a different common information, i.e. while $n_1$, $n_2$ and $n_3$ obtains the same common information, $n_4$, $n_5$ and $n_6$ obtains another. In the second layer, the subnetworks are rearranged as illustrated in Fig.~\ref{subnetworks}. Particularly, a single node from each previous subnetwork is selected into a new subnetwork. As an example, Fig.~\ref{subnetworks} illustrates a subnetwork with $n_4$, $n_4$ and $n_7$ which were in different subnetworks in the first layer. In the second layer, nodes use their common information from the first layer as their input to W-HMAC of the second layer. As a result, each subnetwork obtains the same common information at the end of the second layer, since the second layer W-HMAC inputs are identical for each subnetwork. TPC of the given example is equal to $36$ unit power while TCT is $6$ unit time. TOB of this case is $3$ frequency blocks.

For a $K$-Layer scenario, TPC can be generalized as, $$ N \sum_{k=1}^K (C_k-1),$$ where $C_k, k = 1, \cdots, K$ is the subnetwork size (i.e. number of nodes in a subnetwork) of the $k^{\text{th}}$ layer. TCT is equal to the sum of subnetwork sizes of each layer as $C_1 + C_2 + \dots + C_K$. TOB is the maximum number of subnetworks in a layer as $N/C_k$. It should be noted that reducing the size of the subnetworks to prime factors increases power efficiency while reducing bandwidth efficiency.


\section{Numerical Results} \label{sec:numerical}

The half-duplex key generation method is investigated with simulations in this section. Three realistic scenarios are considered for simulations. 
\begin{itemize}
\item The first scenario (Fig. 4 and Fig. 5): AWGN is added to the received signal of both Eve and the legitimate nodes. Eve is assumed to have the same channel coefficient with the receiving node (as given in Attacker-2 model).
\item The second scenario (Fig. 6): Discrepancies are added to Eve's channel coefficient (as given in Attacker-1 model).
\item The third scenario (Fig. 7): the channel estimation error is added to the transmitting nodes. 
\end{itemize}
The key components are selected from a Gaussian prime set $\mathcal{P} = \{a + bj\}$, $a \in (0,5)$, $b \in (0, 5)$ and $a,b\in \mathbb{N}$. As a benchmark model, a simple broadcasting-based key generation method is considered where each node sequentially broadcast its key components. After $N$ sequential transmission, each node computes the product of obtained key components to construct the common information. Simulations are conducted in MATLAB with 1000 iterations for each figure. Also, MATLAB's VPA (variable-precision floating-point arithmetic) function is used to increase the number of digits evaluated with each function in order to imitate the channel with high precision. The simulation parameters are presented in Table~\ref{tab:params}.

\begin{table}[t]
        \caption{Simulation parameters and definitions}
    \centering \footnotesize
    \begin{tabular}{p{5.5cm}|p{1.5cm}}
\textbf{Parameter} &  \textbf{Quantity} \\ \hline
Number of iterations  & $1000$  \\ \hline 
Message set size (Gaussian primes) &  $8$  \\ \hline
$E_b/N_0$ (dB)  & $20:4:50$   \\ \hline
Number of nodes (Fig. 5) & $3, 4, 5, 6$   \\ \hline
Channel coefficient discrepency ($\sigma$) (Fig. 6) & $0.1, 0.01$
\\ \hline 
Channel estimation error ($\sigma_{\text{estimation}}$) (Fig. 7) & $0.1, 0.01$ \\
    \end{tabular}
    \label{tab:params}
\end{table}

Fig.~\ref{MSE} and Fig.~\ref{errprob1} consider the first scenario where only Gaussian noise is added to the system. In other words, Eve's channel is identical to the legitimate channel (Attacker-2). In Fig.~\ref{MSE}, Mean Squared Errors (MSE) of the obtained keys are presented for $N=3, 5, 7$ networks. Eve is assumed to listen to all communications in the network to generate the secret key since Eve can not obtain the legitimate receiver's key component from a single sniffing. The results show that Eve's key shows more errors than the legitimate node's key. While increasing SNR reduces the MSE of the legitimate node's key, the MSE of Eve's key is not affected. Also, the MSE of both Eve and the legitimate node increase with larger $N$ values since the size of the secret key increases for larger networks.

\begin{figure}[t]
	\includegraphics[width=\linewidth]{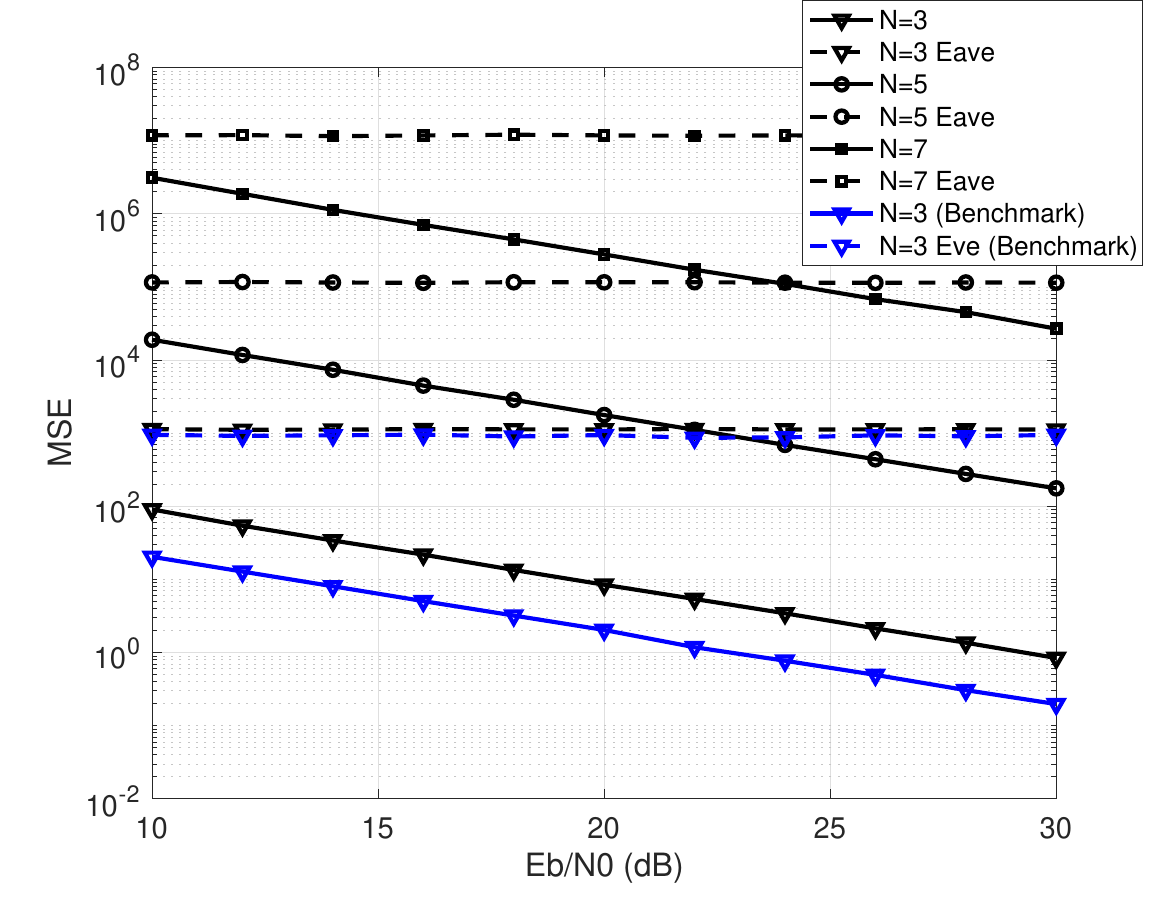}
	\caption{MSE of Eve's and the legitimate node's secret key (TPC$=6, 20, 42$ and TPC$=3, 5, 7$ for $N=3,5,7$, respectively; TOB$=1$ for all network sizes).}
	\label{MSE}
\end{figure}

In Fig.~\ref{errprob1}, the received signals are decoded to the closest possible secret key. The figure illustrates the error probability that the obtained key is different from the actual key. Eve is assumed to generate the secret key by sniffing all communications. As seen from the figure, Eve obtains mismatched (wrong) secret keys with more than $80 \% $ probability for $N=3$ and $95 \% $ probability for $N>3$. However, legitimate nodes can successfully create secret keys with more than $95 \% $ probability at high SNR regions. It should be noted that increasing the number of nodes increases the required SNR level for a feasible success rate. In fact, simultaneously transmitting with more than 5 nodes is not practical in average SNR regions. However, using multiple layers allows us to maintain low error rates (in addition to controlling time, power, and bandwidth consumption). For this reason, we limit our investigation to small network sizes and remind that large networks can be supported with multiple layers using Algorithm \ref{alg2}. In both Fig.~\ref{MSE} and Fig.~\ref{errprob1}, the broadcasting-based benchmark model shows similar results to the proposed model. In fact, it can be seen that the error probability of legitimate communication is smaller for the benchmark model. The reason comes from the fact that Eve's channel is identical to the main channel in Attacker-2 model. This unrealistic assumption is considered only to prove the feasibility of the proposed model under extreme attack conditions. However, this assumption also disables the main advantage of our model against the benchmark model.

\begin{figure}[t]
    \includegraphics[width=\linewidth]{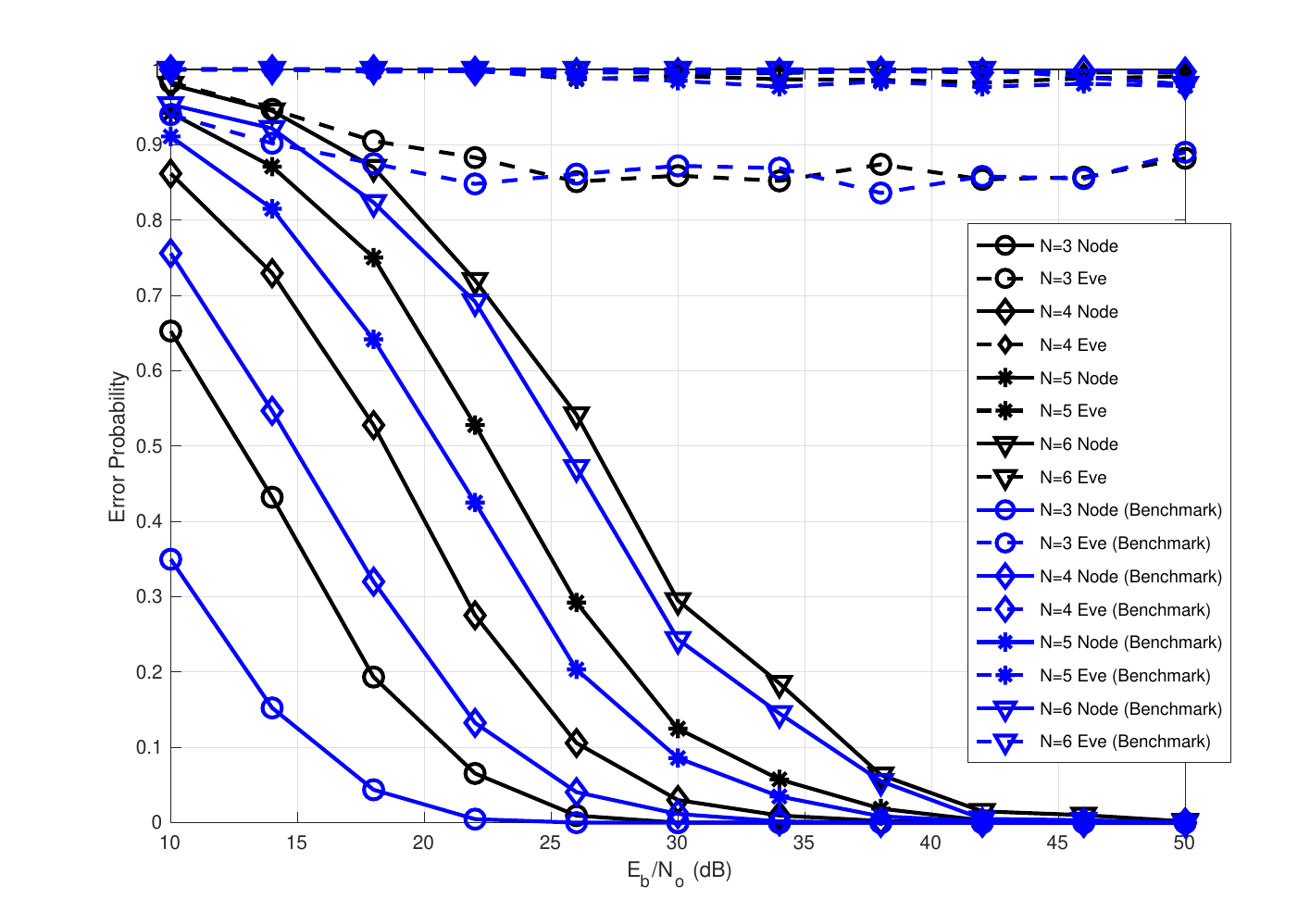}
    \caption{Error probabilities of Eve and the legitimate nodes on the secret key generation with AWGN (TPC$=6, 12, 20, 30$ and TCT$=3, 4, 5, 6$ for $N=3, 4, 5, 6$, respectively; TOB$=1$ for all network sizes).}
    \label{errprob1}
\end{figure}

Fig.~\ref{errprob2} considers the second scenario where discrepancies are added to Eve's channel coefficient (Attacker-1). In the figure, the error probability of Eve and the legitimate nodes is illustrated for $N=3$ networks. Eve is assumed to know the receiver's key component and it creates the secret key from a single sniffing since the scenario aims to investigate the effect of the channel coefficient. When a discrepancy with zero-mean $\sigma=0.1$ Gaussian distribution is added, Eve stays above of $50 \% $ error level. However, when the energy of the discrepancy is reduced to $\sigma=0.01$, the error floor of Eve reduces to $10 \% $. It should be noted that Eve can correctly obtain messages with a higher probability in the benchmark model. On the other hand, the proposed model is more robust against eavesdroppers.

\begin{figure} 
   \includegraphics[width=\linewidth]{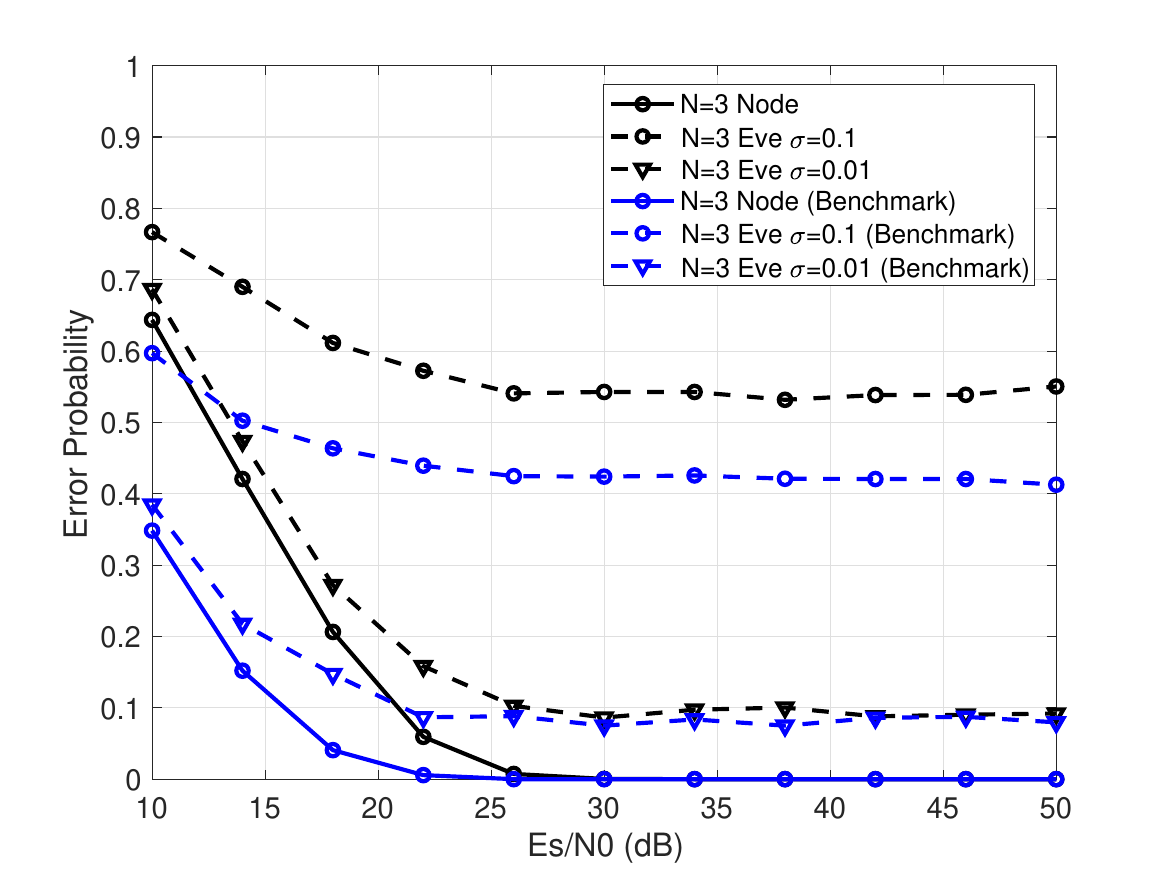}  \vspace{.5 cm}
    \caption{Error probabilities of Eve and the legitimate nodes on the secret key generation with channel coefficient discrepancy (TPC$=6$, TCT$=3$ and TOB$=1$ for all network sizes).}
    \label{errprob2}
\end{figure}

The channel estimation error is added to the system in the third scenario. Fig.~\ref{errprob3} shows the error probability of Eve and the legitimate receiver for $N=3$. The legitimate transmitters are assumed to make a channel estimation error that is modeled with zero-mean $\sigma_{estimation}$ Gaussian distribution. AWGN noise is considered in to all receivers and Eve is assumed to have $\sigma=0.1$ discrepancy on its channel coefficient. Also, Eve is assumed to create the secret key from a single sniffing since the scenario aims to investigate the effect of the channel estimation error. When $\sigma_{estimation}=0.01$ is considered, error probability of the legitimate nodes reduces below $5 \%$ after $25$ dB SNR. However, the feasibility of the proposed model is badly affected when $\sigma_{estimation}=0.1$, since the error probability of legitimate nodes rises to $50 \%$ floor. 

\begin{figure} [t]
\includegraphics[width=\linewidth]{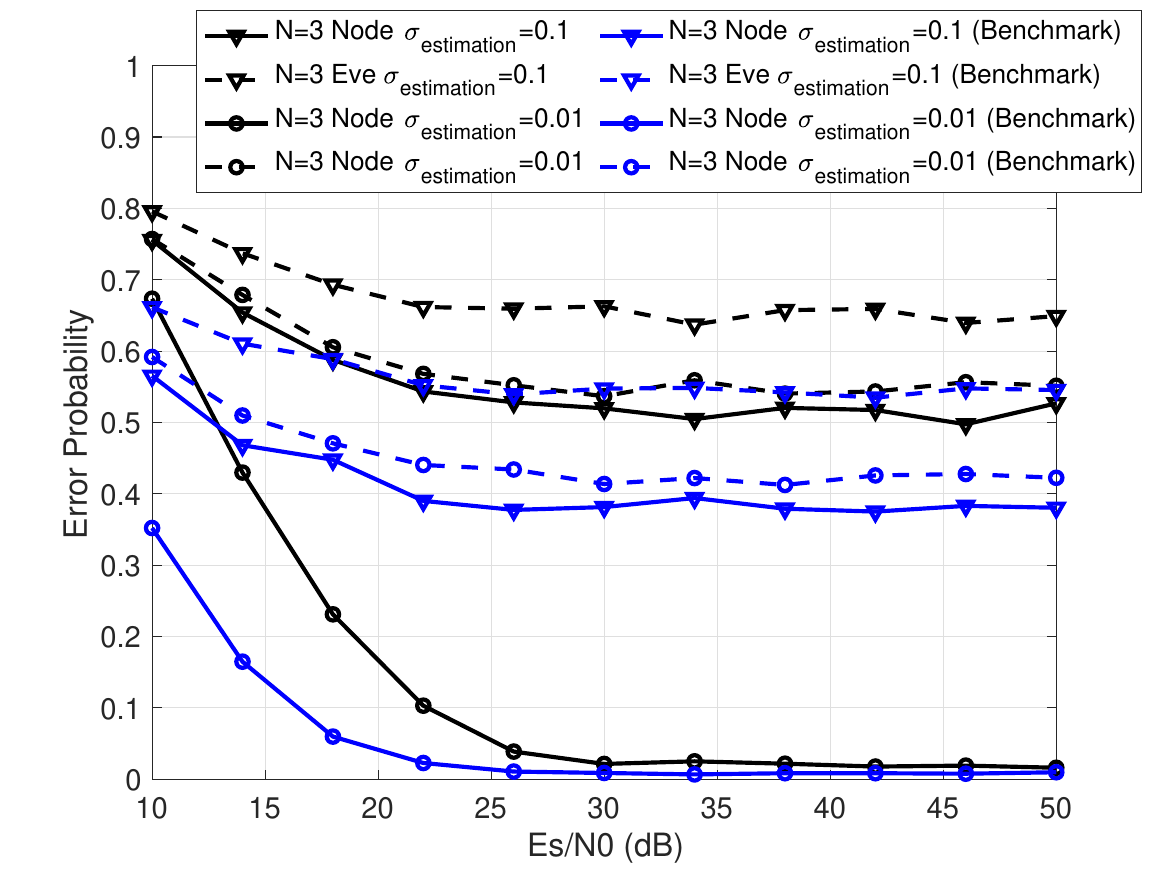} 
\caption{Error probabilities of Eve and the legitimate nodes on the secret key generation with channel estimation error (TPC$=6$, TCT$=3$ and TOB$=1$ for all network sizes).}
\label{errprob3}
\end{figure}

\begin{table*}[t]  \caption{TPC, TCT and TOB comparison of various state-of-the-art models and the proposed multi-layer approach.} 
	\centering 
	\begin{tabular}{ |c|c|c|c|c|c|c|c| }

		\hline
		\multicolumn{8}{|c|}{  \textbf{The State-of-the-Art Models} }  \\ \hline
		
	\multicolumn{4}{|c|}{  \textbf{The paper}} &  \textbf{WSN Size} & \textbf{TPC} & \textbf{TCT} & \textbf{TOB} \\ \hline
				
    \multicolumn{4}{|c|}{ Liu \textit{et al. \cite{Liu2014}  }} &  $N$  &  $N(N-1)$  &  $N$  &  $1$  \\ \hline
		
    \multicolumn{4}{|c|}{ Xiao \textit{et al.} \cite{Xiao2018}}  &  $N$  &  $N(N-1)$  &  $N$  &  $1$  \\ \hline
	
    \multicolumn{4}{|c|}{ Thai \textit{et al.}    \cite{Thai2019} }  &  $N$  &  $N(N-1)$  &  $2N$  &  $1$  \\ \hline

    \multicolumn{4}{|c|}{ Li \textit{et al.}  \cite{Guyue2019} } &  $N$  &  $N$  &  $N$  &  $1$  \\ \hline

    \multicolumn{4}{|c|}{ Zhang \textit{et al.} \cite{Zhang2019} } &  $N$  &  $N(N-1)$  &  $N$  &  $1$  \\

		\hline \hline
		\multicolumn{8}{|c|}{\textbf{The Proposed Model}}\\ \hline
		
		\textbf{WSN Size} & \textbf{Layer size} &  $\mathbf{1^{\textbf{st}}}$ \textbf{layer} & $\mathbf{2^{\textbf{nd}}}$ \textbf{layer}  & $\mathbf{3^{\textbf{rd}}}$ \textbf{layer} & \textbf{TPC} & \textbf{TCT} & \textbf{TOB} \\
		\hline
		
		$N$ & $1$-Layer  &  $N$ & & & $N (N-1)$ & $N$ & $1$ \\ \hline    
		$N=C_1 C_2 \dots C_K$ & $K$-Layer & $C_1$ & $C_2$ & $C_3$ & $N \sum_{k=1}^K (C_k-1)$ & $C_1 + C_2 + \dots + C_K$ & $\max[N/C_k]$ \\ \hline \hline
		$9$ & $1$-Layer &  $9$ & & & $72$ & $9$ & $1$ \\ \hline
		$9$ & $2$-Layer &  $3$ & $3$ & & $36$ & $6$ & $3$ \\ \hline
		$12$ & $1$-Layer & $12$ & & & $132$ & $12$ & $1$ \\ \hline
		$12$ & $2$-Layer & $3$ & $4$ & & $60$ & $7$ & $4$ \\ \hline
		$12$ & $3$-Layer & $2$ & $2$ & $3$ & $48$ & $7$ & $6$ \\ \hline
		$14$ & $2$-Layer & $2$ & $7$ & & $98$ & $9$ & $7$ \\ \hline
		$27$ & $3$-Layer & $3$ & $3$ & $3$ & $162$ & $9$ & $9$ \\ \hline
	\end{tabular}
	\label{tab:sizes}
\end{table*}

We have discussed error probability results of the proposed secret key generation scheme over W-HMAC by considering the given network in a single layer structure. In Section \ref{skgmultilayer}, we proposed TPC, TCT and TOB metrics to verify the performance of multiple layer structure presented in Algorithm~\ref{alg2}. Herein, we compare TPC, TCT and TOB performance of various state-of-the-art models and our approach. In Table~\ref{tab:sizes}, a number of possible layer structures and their time, energy and bandwidth costs are presented without including channel probing or estimation load. The proposed single-layer approach shows identical performance compared to most of the state-of-the-art models. However, multiple layer structure of the proposed model enables us to improve a performance metric by leveraging multiple-layers. For instance, Table~\ref{tab:sizes} presents the TPC, TCT and TOB of $N=9$ networks for $K=1,2,3$ layers.
Comparing 1-layer and 2-layer cases, TPC of the key generation process can be reduced to $36$ (from $72$) by using two layers. Also, TCT is reduced to $6$ (from $9$) in two layer case. However, multiple subnetworks require multiple frequency blocks to support unique W-HMAC configurations. In the given example, TOB is increased from $1$ to $3$ in two layer case. As seen from the example, using multiple layers can provide a trade-off between time, energy and bandwidth requirements of the proposed method.

\section{Conclusion}

We model the W-MAC with proper \textit{pre}-processing and \textit{post}-processing functions in order to match the W-MAC with a function that outputs a secret key which includes information from other users. 
The model is also extended to multiple layers to obtain flexibility between the time and frequency resources.
Lastly, we examine the performance of the proposed model for various  scenarios in a competitive manner with the benchmark system. It is shown that the proposed model provides security against passive eavesdroppers in certain scenarios. As a complementary future work, we consider to investigate and to combat the effect of noise in the proposed methods.

\appendices
\section{Proof of the Lemma~\ref{lemma1}}
    
Let $n=n_1n_2\dots n_d$ where $d$ is the number of digits of $n$. The decimal representation:
$$n=n_1\cdot 10^{d-1}+n_2\cdot 10^{d-2}+\dots+n_{d-1}\cdot 10 + n_d$$ and
$$s=1+a_1\cdot \dfrac{1}{10}+a_2\cdot \dfrac{1}{10^2}+\dots +a_r\dfrac{1}{10^r}+a_{r+1}\dfrac{1}{10^{r+1}}+\dots + a_m\dfrac{1}{10^{M}}$$ as $a_r$ is the first non-zero term in the decimal part $$a_1=a_2=\dots=a_{r-1}=0$$ and
$$s=1+a_r\dfrac{1}{10^r}+a_{r+1}\dfrac{1}{10^{r+1}}+\dots+a_m\dfrac{1}{10^{M}}$$

Then
$$\begin{array}{llll}
n\times s&=&\left( n_1 \times 10^{d-1}+n_2 \times 10^{d-2}+\dots + n_d\right)\cdot\\
&&\left (1+a_1 \dfrac{1}{10^r}+a_2 \dfrac{1}{10^{r+1}}+\dots a_m\dfrac{1}{10^{m}}\right)\\
&=&n_1 \times 10^{d-1}+n_2 \times 10^{d-2}+\dots+ n_d\\
&&+a_r \left(n_1 \times 10^{d-r-1}+\dots+n_d \times 10^{-r} \right)\\
&&+ a_{r+1} \left( n_1\times 10^{d-r-2} +\dots+n_d\times {10^{-(r+1)}}\right) \\
&&+\cdots + a_m \left( n_1\times 10^{d-m-1} +\dots+n_d\times {10^{-m}}\right) \\
&=&n_1\times 10^{d-1}+n_2\times 10^{d-2}+\dots+(n_{d-r-1}+n_1)\times \\
&&10^{d-r-1}+(n_{d-r-2}+n_2+n_1)\times 10^{d-r-2}.
\end{array}$$
Hence multiplication results in the number with the same first $r$ (or $r-1$) digits as the number $n$.

\section*{Acknowledgment}

The authors would like to thank Prof. Halim Yanikomeroglu for his invaluable comments.

\ifCLASSOPTIONcaptionsoff
  \newpage
\fi

\bibliographystyle{IEEEtran}
\bibliography{ref}





\end{document}